\documentclass[11pt]{amsart}
\usepackage{graphicx,amssymb,amsmath,amsthm}
\usepackage{enumerate}
\usepackage{dsfont}
\usepackage[colorlinks, citecolor=red]{hyperref}
\usepackage{comment,cite,color}
\usepackage{cite,color}
\usepackage{mathrsfs}
\usepackage{epsfig}
\usepackage{lscape}
\usepackage{subfigure}
\usepackage{epstopdf}

\usepackage{caption}
\usepackage{algorithm}
\usepackage{algpseudocode}
\numberwithin{equation}{section}

\newcommand{\innerp}[1]{\langle {#1} \rangle}
\newcommand{\norm}[1]{\|{#1}\|_2}
\newcommand{\normf}[1]{\|{#1}\|_F}

\newcommand{\abs}[1]{\lvert#1\rvert}

\newcommand{\floor}[1]{\left\lfloor #1 \right\rfloor}

\newcommand{\A}{{\mathcal A}}
\newcommand{\M}{{\mathbf M}}

\newcommand{\R}{{\mathbb R}}
\newcommand{\Rd}{{\mathbb R}^d}

\newcommand{\T}{\top}
\newcommand{\C}{{\mathbb C}}
\newcommand{\Cd}{{\mathbb C}^d}
\newcommand{\x}{{\tilde{\mathbf{x}}}}
\newcommand{\y}{{\tilde{\mathbf{y}}}}
\newcommand{\Real}{{\mathfrak{R}}}
\newcommand{\uu}{{\tilde{\mathbf{u}}}}
\newcommand{\uv}{{\tilde{\mathbf{v}}}}
\newcommand{\vx}{{\mathbf x}}
\newcommand{\vy}{{\mathbf y}}
\newcommand{\vu}{{\mathbf u}}
\newcommand{\vv}{{\mathbf v}}
\newcommand{\vz}{{\mathbf z}}
\newcommand{\vb}{{\mathbf b}}
\newcommand{\ve}{{\mathbf e}}
\newcommand{\va}{{\mathbf a}}

\newcommand{\F}{{\mathbb F}}
\newcommand{\Z}{{\mathbb Z}}

\newcommand{\rank}{{\rm rank}}
\newcommand{\tr}{{\rm tr}}

\newcommand{\spann}{{\rm span}}
\newcommand{\dimm}{{\rm dim}}

\newtheorem{definition}{Definition}[section]
\newtheorem{corollary}[definition]{Corollary}

\newtheorem{theorem}[definition]{Theorem}
\newtheorem{lemma}[definition]{Lemma}

\newtheorem{remark}[definition]{Remark}

\date{}

\begin{document}
\baselineskip 18pt
\bibliographystyle{plain}
\title{
Phase retrieval  from the norms of affine  transformations 
}

\author{Meng Huang}
\address{LSEC, Inst.~Comp.~Math., Academy of
Mathematics and System Science,  Chinese Academy of Sciences, Beijing, 100091, China}
\email{hm@lsec.cc.ac.cn}

\author{Zhiqiang Xu}
\thanks{Zhiqiang Xu was supported  by NSFC grant (11422113, 91630203, 11331012) and by National Basic Research Program of China (973 Program 2015CB856000)}
\address{LSEC, Inst.~Comp.~Math., Academy of
Mathematics and System Science,  Chinese Academy of Sciences, Beijing, 100091, China
\newline
School of Mathematical Sciences, University of Chinese Academy of Sciences, Beijing 100049, China}
\email{xuzq@lsec.cc.ac.cn}

\begin{abstract}
In this paper, we consider the generalized phase retrieval from affine measurements. This problem aims to recover  signals $\vx \in \F^d$ from the magnitude of the affine transformations  $y_j=\norm{M_j^*\vx +\vb_j}^2,\; j=1,\ldots,m,$ where $M_j \in \F^{d\times r}, \vb_j\in \F^{r}, \F\in \{\R,\C\}$ and we call it as  {\em generalized affine phase retrieval}. We develop a framework
for generalized affine phase retrieval  with presenting  necessary and sufficient conditions for
$\{(M_j,\vb_j)\}_{j=1}^m$ having generalized affine phase retrieval property.
We also establish results on minimal measurement number for generalized affine phase retrieval.
 Particularly, we show if $\{(M_j,\vb_j)\}_{j=1}^m \subset \F^{d\times r}\times \F^{r}$
 has generalized affine phase retrieval property, then $m\geq d+\floor{d/r}$ for $\F=\R$
 ($m\geq 2d+\floor{d/r}$ for $\F=\C$ ). We also show that the bound is tight provided $r\mid d$.
  These results imply that one can reduce the measurement number by raising $r$, i.e. the rank of $M_j$. This highlights a notable difference between  generalized affine phase retrieval and generalized phase retrieval.
 Furthermore, using tools of algebraic geometry, we show that $m\ge 2d$ (resp. $m\ge 4d-1$) generic measurements $\A=\{(M_j,b_j)\}_{j=1}^m$ have the generalized phase retrieval property for $\F=\R$ (resp. $\F=\C$).
\end{abstract}
\maketitle
\section{Introduction}
\subsection{Phase retrieval}
 Phase retrieval aims to recover a signal $\vx\in \F^d$ from the measurements $\abs{\innerp{\mathbf{a}_j,\vx}},\;j=1,\ldots,m $, where $\F=\R$ or $\mathbb{C}$ and $\mathbf{a}_j\in \F^d$ are the measurement vectors. Phase retrieval  is raised in many areas  such as X-ray crystallography \cite{harrison1993phase,millane1990phase}, microscopy \cite{miao2008extending}, astronomy \cite{fienup1987phase}, coherent diffractive imaging \cite{shechtman2015phase,gerchberg1972practical} and optics \cite{walther1963question}.
To state conveniently, set $A:=(\mathbf{a}_1,\ldots,\mathbf{a}_m)$ and $\M_A(\vx):=(\abs{\innerp{\mathbf{a}_1,\vx}},\ldots,\abs{\innerp{\mathbf{a}_m,\vx}})\in \R^m$. Noting that for any $c\in \F$ with $\abs{c}=1$  we have $\M_A(\vx)=\M_A(c\vx)$ and hence we can only hope to recover $\vx$ up to a unimodular constant.
If $\M_A(\vx)=\M_A(\vy)$ implies $\vx\in \{c\vy:c\in \F, \abs{c}=1\}$, we say $A$ has {\em phase retrieval property for $\F^d$}. A fundamental  problem in phase retrieval is to give the minimal $m$ for which there exists $A=(\mathbf{a}_1,\ldots, \mathbf{a}_m)^\T\in \F^{m\times d}$ which has phase retrieval property for $\F^d$.  For the case $\F=\R $, it is well known that the minimal measurement number  $m$ is  $2d-1$ \cite{balan2006signal}.  For the complex case $\F=\mathbb{C}$, this question remains open. Conca, Edidin, Hering and Vinzant \cite{conca2015algebraic}  proved $m\ge 4d-4 $ generic measurement vectors $A=(\mathbf{a}_1,\ldots,\mathbf{a}_m)^\T\in \C^{m\times d}$ have phase retrieval property for $\C^d$ and they furthermore show that $4d-4$  is sharp  if $d$ is in the form of $2^k+1,\;k\in \mathbb{Z}_+$. In \cite{vinzant2015small}, for the case where $\F=\C$ and $d=4$, Vinzant  present $11=4d-5<4d-4$ measurement vectors which has phase retrieval property for $\mathbb{C}^4$ which  implies that $4d-4$ is not sharp for some dimension $d$. Beyond the minimal measurement number problem, one also developed many efficient algorithms for recovering $\vx$ from $\M_A(\vx)$ (see \cite{phase1,phase2,phase3}).
\subsection{Generalized phase retrieval and affine phase retrieval}
A generalized version of phase retrieval, which is called  as {\em generalized phase retrieval}, was introduced by Wang and Xu \cite{wang2016generalized}.
 In the generalized phase retrieval, one aims to reconstruct $\vx\in \F^d$ through quadratic samples $\vx^*A_1\vx,\ldots,\vx^*A_m\vx$ where $A_j\in \F^{d\times d}$ are Hermitian matrix  for $\F=\C$ (symmetric matrix for $\F=\R$). Set $\A:=(A_j)_{j=1}^m$ and $\M_\A(\vx):=(\vx^*A_1\vx,\ldots,\vx^*A_m\vx)$. We say $\A$ has {\em generalized phase retrieval property} if $\M_\A(\vx)=\M_\A(\vy)$ implies that $\vx\in \{c\vy:c\in \F,\abs{c}=1\}$.
 In \cite{wang2016generalized}, Wang and Xu show the fantastic connection among phase retrieval, nonsingular bilinear form and embedding. They also study the minimal $m$ for which there exists $\A=(A_j)_{j=1}^m$ which has generalized phase retrieval property. Particularly, they show that for the case $\F=\C$, the measurement number $m\geq 4d-2-2\alpha$ where $\alpha$
denotes the number of 1's in the binary expansion of $d-1$.
If take $A_j=\va_j\va_j^*$, then the generalized phase retrieval is reduced to the standard phase retrieval. Furthermore, if we require $A_j, j=1,\ldots,m,$ are orthogonal projection matrices,  the generalized phase retrieval is reduced to phase retrieval by projection \cite{projphase,planephsae}. Hence, the generalized phase retrieval  includes the standard phase retrieval as well as the phase retrieval by projection as a special case.
Both  standard phase retrieval and generalized phase retrieval require the measurement number is greater than or equal to $4d-2-2\alpha$. Hence, one can not reduce the minimal measurement number heavily  by rasing the rank of $A_j$.

Affine phase retrieval is raised in holography \cite{liebling2003local} as well as in phase retrieval with background information \cite{YW} which aims to recover $\vx\in \F^d$ from $\abs{\innerp{\mathbf{a}_j,\vx}+b_j},\;j=1,\ldots,m,$ where $\mathbf{a}_j\in \F^d$ and $b_j\in \F$.  The authors of \cite{gao2016phase} develop the general framework of affine phase retrieval with highlighting the difference between affine phase retrieval and standard phase retrieval.
Unlike the standard phase retrieval where we can only recover $\vx$ up to a unimodular constant,  it is possible to recover $\vx$ exactly in affine phase retrieval.
 Particularly, for the case where $\F=\C$, the authors of \cite{gao2016phase} show that there exist $m=3d$ measurements $\{(\mathbf{a}_j,b_j)\}_{j=1}^m$ so that  one can recover $\vx$ from $\abs{\innerp{\mathbf{a}_j,\vx}+b_j}, j=1,\ldots,m$. They furthermore show the measurement number $3d$ is sharp for recovering $\vx\in \C^d$ from $\abs{\innerp{\mathbf{a}_j,\vx}+b_j}, j=1,\ldots,m$.  Similarly, for the case where $\F=\R$, if was shown in \cite{gao2016phase} that $m=2d$ measurements are sufficient and necessary for recovering $\vx$ from $\abs{\innerp{\mathbf{a}_j,\vx}+b_j}, j=1,\ldots,m$.
\subsection{Generalized affine phase retrieval}
In this paper, we consider the recovery of $\vx\in \F^d$ from the affine quadratic measurements
\[
y_j=\norm{M_j^* \vx +\vb_j}^2,\quad j=1,\ldots,m ,
\]
where $M_j \in \F^{d\times r}$ and $\vb_j\in \F^{r}$.  Set $\A=\{(M_j,\vb_j)\}_{j=1}^m \subset \F^{d\times r}\times \F^{r}$, we can view $\A$ as a point in $\F^{m(d\times r)}\times \F^{mr}$.  Define the map $\M_{\A}:\F^d\rightarrow \R^m$ by
\begin{equation}\label{map}
  \M_{\A}(\vx)=(\norm{M_1^* \vx+\vb_1}^2,\ldots,\norm{M_m^* \vx+\vb_m}^2).
\end{equation}
Our aim  is to study whether a signal $\vx\in \F^d$ can be uniquely reconstructed from $\M_{\A}(\vx)$.
 To state conveniently, we introduce the definition of {\em the generalized affine phase retrieval property}.
\begin{definition}
Let $r\in \Z_{\geq 1}$ and $\A=\{(M_j,\vb_j)\}_{j=1}^m \subset \F^{d\times r}\times \F^{r}$. We say $\A$ has the generalized affine phase retrieval property if  $\M_{\A}$ is injective on $\F^d$.
\end{definition}
We next introduce the connection between generalized affine phase retrieval and generalized phase retrieval. Note that
\begin{equation}\label{eq:con}
y_j=\norm{M_j^*\vx +\vb_j}^2=\x^* A_j \x, \; j=1,\ldots,m ,
\end{equation}
where
\[ \x=\left(\begin{array}{c}
              \vx \\
              1
            \end{array}
\right) \qquad \text{and} \qquad A_j=\left(
                                     \begin{array}{cc}
                                       M_jM_j^* & M_j\vb_j \\
                                      (M_j\vb_j)^*  & \vb_j^*\vb_j \\
                                     \end{array}
                                   \right).
\]
The (\ref{eq:con}) shows that generalized affine phase retrieval can be reduced  to  recover  $\x\in \F^{d+1}$ from $\x^*A_j\x, j=1,\ldots,m$.
Since we already know the last entry of $\x$ is $1$, we can recover $\x$ from $\x^*A_j\x, j=1,\ldots,m$ exactly.  Hence,  the generalized affine phase retrieval can be considered as the extension of both the generalized phase retrieval and the affine phase retrieval.

\subsection{Continuous map}

Note that $\vx\in \R^d$ has $d$ real variables ($2d$ real variables for the complex case). Naturally, one may be interested in whether it is possible to recover $\vx\in \R^d$ from $d$ nonnegative  measurements ($2d$ nonnegative measurements for $\F=\C$).
We state the question as follows. For $j=1,\ldots,m$, suppose that $f_j: \F^d\rightarrow \R_+$ is a continuous nonnegative function, i.e. $f_j(\vx)\geq 0$. For $\vx\in \F^d$, set
\begin{equation}\label{eq:Fdef}
{\mathbf F}(\vx):=(f_1(\vx),\ldots,f_m(\vx))\in \R_+^m.
\end{equation}
One may be interested in the question:{\em
What is the smallest $m$ so that ${\mathbf F}$ is injective on $\R^d$?} Under some mild conditions for ${\mathbf F}$, we show that $m\geq d+1$ is necessary for ${\mathbf F}$ being injective on $\Rd$ ($m\geq 2d+1$ for $\C^d$). As we will show later, there exists $\{(A_j,\vb_j)\}_{j=1}^m\subset \R^{d\times d}\times \R^d$ with $m=d+1$ so that
$ \M_{\A}$ is injective on $\Rd$. This implies that the generalized affine phase retrieval can achieve the lower bound $m=d+1$. A similar conclusion also holds for the complex case.

\subsection{Our contribution}
In this paper, we  develop the framework of the  generalized affine phase retrieval.
 Particularly, we focus  on the number of measurements needed to
achieve generalized affine phase retrieval.
 We first present some equivalent conditions  and then study the minimal measurement number to guarantee the generalized affine phase retrieval property for both real and complex signals. For $\F=\R$, we show that $ m\ge d+\floor{\frac{d}{r}}$  ($m\ge 2d+\floor{\frac{d}{r}}$ for $\F=\C$) is necessary for there existing measurements $\{(M_j,\vb_j)\}_{j=1}^m\subset \F^{d\times r}\times \F^r$ which have this property. We also show that the bound is tight provided $d/r\in \Z$.
Compared with the generalized  phase retrieval, the generalized affine phase retrieval can reduce the measurement number heavily by rasing the rank of $M_j$.  This also highlights a notable difference between the generalized affine phase retrieval and generalized phase retrieval.

 Using the tools developed in \cite{balan2006signal, conca2015algebraic, wang2016generalized},
 we show that $m\ge 2d$  generic measurements $\{(M_1,\vb_1),\ldots,(M_m,\vb_m)\}\in \F^{m(d\times r)\times mr}$ for $\F=\R$ ($m\ge 4d-1$ for $\F=\C$) can do generalized affine phase retrieval for $\F^d$.

\section{The minimal measurement number for Continuous map}
Recall that ${\mathbf F}: \F^d\rightarrow \R_+^m$ is a continuous map.
The next theorem shows that the necessary condition for ${\mathbf F}$ being injective is $m\geq d+1$ under some mild condition for ${\mathbf F}(\vx)$.
\begin{theorem}\label{th:d1}
Suppose that ${\mathbf F}:\F^d\rightarrow \R_+^m$ is a continuous map which satisfies
\begin{equation}\label{eq:wuqiong}
\lim_{R\rightarrow +\infty}{\rm inf}_{\|\vx\|\geq R} \|{\mathbf F}(\vx)\|=+\infty.
\end{equation}
Then if $m=d$ and $\F=\R$ ($m=2d$ for $\F=\C$), then ${\mathbf F}$ is not injective on $\F^d$.
\end{theorem}
\begin{proof}
Note that $\C^d\cong \R^{2d}$. We just need consider the case where $\F=\R$.
To this end, we use
$\mathbb{S}^d$ to denote the $d$-sphere and use ${\mathcal N}$ to denote the north pole
of $\mathbb{S}^d$. Let $g: \Rd \rightarrow \mathbb{S}^d\setminus \{ {\mathcal N}\}$ be the natural homeomorphism between $\Rd$ and $\mathbb{S}^d\setminus \{ {\mathcal N}\}$.
 Then ${\mathbf F}_g:=g\circ {\mathbf F}\circ g^{-1}$ is the operator which maps $\mathbb{S}^d\setminus \{{\mathcal N}\}$ to $g(\Rd_+)\subset \mathbb{S}^d\setminus \{{\mathcal N}\}$.
Set
 \begin{equation*}
 \widetilde{\mathbf F}_g(\vx):=\left\{
        \begin{array}{cl}
          {\mathbf F}_g(\vx), & \vx\in \mathbb{S}^d\setminus \{ {\mathcal N}\}\\
          \vx, & \vx={\mathcal N}
        \end{array}
      \right. .
\end{equation*}
Since ${\mathbf F}$ satisfies (\ref{eq:wuqiong}),  $\tilde{\mathbf F}_g$  is continuous
on $\mathbb{S}^d$.
 Note that $\widetilde{\mathbf F}_g(\mathbb{S}^d)\subset g(\Rd_+)\cup \{{\mathcal N}\}$. Thus the range of $\widetilde{\mathbf F}_g$ is not the whole $\mathbb{S}^d$, which means that $\widetilde{\mathbf F}_g(\mathbb{S}^d)\hookrightarrow \Rd$. We now get a continuous map from
$\mathbb{S}^d$ to $\Rd$ and we abuse the notation and still use $\widetilde{\mathbf F}_g$ to denote the map. By Borsuk-Ulam theorem, there exists $\{\vx,-\vx\}\subset \mathbb{S}^d$ such that $\widetilde{\mathbf F}_g(\vx)=\widetilde{\mathbf F}_g(-\vx)$. Let $\vy_1=g^{-1}(\vx)$ and $\vy_2=g^{-1}(-\vx)$, and then ${\mathbf F}(\vy_1)={\mathbf F}(\vy_2)$ since $g$ is injective. Now, we claim that $\vy_1\neq \infty$ and $\vy_2\neq \infty$. Indeed, if $\vy_1= \infty$, then $\vx={\mathcal N}$ since $\vx=g(\vy_1)$.
 Hence $-\vx$ is the south pole which implies that ${\mathbf F} (\vy_2)$ is finite since $\vy_2=g^{-1}(-\vx)$.
Hence, we find two points $\vy_1\neq \vy_2\in \Rd$, but ${\mathbf F}(\vy_1)={\mathbf F}(\vy_2)$, which arrives at the conclusion.
\end{proof}
\begin{remark}
In Theorem \ref{th:d1}, we require that the image of ${\mathbf F}=(f_1,\ldots,f_m)$ is a subset of ${\mathbb R}_+^d$. If we remove the requirement of $f_j(\vx)\geq 0$, then there exists a map ${\mathbf F}:\Rd\rightarrow \Rd$ which is injective on $\R^d$. In fact, we just take ${\mathbf F}(\vx)=(\innerp{\va_1,\vx},\ldots,\innerp{\va_d,\vx})$ where $\va_j\in \R^d$ satisfying  ${\rm span}\{\va_1,\ldots,\va_d\}=\R^d$, and then ${\mathbf F}$ is injective on $\R^d$.
Moreover, if we remove the condition (\ref{eq:wuqiong}), we can set ${\mathbf F}(\vx):=(\exp(x_1),\ldots,\exp(x_d))$ which is also injective on $\R^d$.
\end{remark}
\section{Generalized affine phase retrieval for real signals}
In this section, we consider the generalized affine phase retrieval for  real signals. We first state several  equivalent conditions  for the generalized affine phase retrieval.
Suppose that $M\in \R^{d\times r}$ and $\vb\in \R^r$. Then the following formula is straightforward to check:
\begin{equation}\label{basic_formula}
  \norm{M^\T \vx +\vb}^2-\norm{M^\T \vy +\vb}^2=4\left(\vu^\T MM^\T \vv+(M\vb)^\T \vv\right)\; \text{ for any } \vx,\vy\in \Rd
\end{equation}
where $\vu=\frac{1}{2}(\vx+\vy)$ and $\vv=\frac{1}{2}(\vx-\vy)$.
\begin{theorem}\label{equivelant in R}
Suppose that $r\in \Z_{\geq 1}$.
Let $\A=\{(M_j,\vb_j)\}_{j=1}^m \subset \R^{d\times r}\times \R^{r}$. Then the followings are equivalent:
\begin{itemize}
  \item[(1)] $\A$ has the generalize affine phase retrieval property for $\Rd$.
  \item[(2)] For any $\vu,\vv\in \Rd$ and $\vv\neq 0$, there exists a $j$ with $1\le j\le m$ such that $$\vu^\T M_jM_j^\T \vv+(M_j\vb_j)^\T \vv\neq 0.$$
  \item[(3)]  $ \spann\{M_jM_j^\T \vu+M_j\vb_j\}_{j=1}^m=\Rd$ for any $\vu\in \Rd$.
  \item[(4)] The Jacobian of $\M_{\A}$ has rank $d$ for all $\vx\in \Rd$.
\end{itemize}
\end{theorem}
\begin{proof}
(1)$\Leftrightarrow$(2). Assume that there exist $\vx\neq \vy$ in $\Rd$ such that $\M_{\A}(\vx)-\M_{\A}(\vy)=0$. Then from (\ref{basic_formula}) for all $j$ we have
\[ \norm{M_j^\T \vx +\vb_j}^2-\norm{M_j^\T \vy +\vb_j}^2=4(\vu^\T M_jM_j^\T \vv+(M_j\vb_j)^\T \vv)=0.
\]
Note that $\vv\neq 0$ and then we conclude a contradiction with (2). It means that (2) $\Rightarrow$ (1). The converse also follows from the same argument.

(2)$\Leftrightarrow$(3). If for some $\vu$ such that $ \spann\{M_jM_j^\T \vu+M_j\vb_j\}_{j=1}^m\neq\Rd$, then there exists  a $\vv\neq 0$ such that $\vv\perp \spann\{M_jM_j^\T \vu+M_j\vb_j\}_{j=1}^m$. It implies that $\vu^\T M_jM_j^\T \vv+(M_j\vb_j)^\T \vv=0$ for all $j=1,\ldots,m$. This is a contradiction. The converse clearly also holds.

(3)$\Leftrightarrow$(4). Note that the Jacobian $J(\vx)$ of the map $\M_{\A}$ at $\vx\in \Rd$ is exactly
\[ J(\vx)=2[M_1M_1^\T \vx+M_1\vb_1,\ldots,M_mM_m^\T \vx+M_m\vb_m].
\]
Thus (3) is equivalent to that the rank of $J(\vx)$ is $d$ for all $\vx\in \Rd$.
\end{proof}

\begin{corollary}
Suppose that $r\in \Z_{\geq 1}$ and $\A=\{(M_j,\vb_j)\}_{j=1}^m$ where $ (M_j,\vb_j)\in \R^{d\times r}\times \R^{r}$. If $\A$ has generalized affine phase retrieval property for $\R^d$ then
 $m\geq d+\floor{\frac{d}{r}}$.
\end{corollary}
\begin{proof}
To this end, we just need show that if
$m\le d+\floor{\frac{d}{r}}-1$, then $\A$ is not generalized affine phase retrievable for $\Rd$.
When $r\geq d+1$, the conclusion follows from (3) in Theorem \ref{equivelant in R} directly. Hence, we only consider the case where $r\leq d$.
A simple observation is that there exists $\vu\in \Rd$ such that $M_j^\T \vu+\vb_j=0, j=1,\ldots,\floor{\frac{d}{r}}$.
 Thus, if $m\le \floor{\frac{d}{r}}+ d-1$ then
\[
\spann\{M_jM_j^\T \vu+M_j\vb_j\}_{j=1}^m=\spann\{M_jM_j^\T \vu+M_j\vb_j\}_{j=\floor{\frac{d}{r}}+1}^m\neq \Rd.
\]
According to (3) in Theorem \ref{equivelant in R}, we arrive at the conclusion.
\end{proof}

According to the above corollary, if $\{(A_j,\vb_j)\}_{j=1}^m$ is generalized affine phase retrievable for $\Rd$ then $m\geq d+\floor{\frac{d}{r}}$. We next show the bound $d+\floor{\frac{d}{r}}$ is tight provided $r \mid d$. To this end, we introduce the following lemma:
\begin{lemma} \label{lemma in R}
Suppose that $\vb_1,\ldots,\vb_{r+1} \in \R^r$ satisfy
\begin{equation}\label{eq:span}
{\rm span}\{\vb_2-\vb_1,\vb_3-\vb_1,\ldots,\vb_{r+1}-\vb_1\}=\R^r.
 \end{equation}
 Then $\vx=\vy$ if and only if $\norm{\vx+\vb_j}=\norm{\vy+\vb_j}$ for all $j=1,\ldots,r+1$ where $\vx,\vy\in \R^r$.
\end{lemma}
\begin{proof}
We denote $\vz:=\vx-\vy \in \R^r$ and $t:=(\norm{\vx}^2-\norm{\vy}^2)/2$. Then $\norm{\vx+\vb_j}=\norm{\vy+\vb_j}$ is equivalent to $\vb_j^\T \vz+t=0$ for all $j=1,\ldots,r+1$.
To this end, we just need show that $\norm{\vx+\vb_j}=\norm{\vy+\vb_j}$ for all $j=1,\ldots,r+1$ implies $\vx=\vy$.
According to (\ref{eq:span}), the linear system
 \[ \left(
      \begin{array}{cc}
        \vb_1^\T & 1 \\
         \vdots & \vdots \\
        \vb_{r+1}^\T & 1 \\
      \end{array}
    \right)\left(
             \begin{array}{c}
               \vz \\
               t \\
             \end{array}
           \right)=0
 \]
has only zero solution, i.e., $(\vz,t)=0$, which implies  $\vx=\vy$.
\end{proof}

\begin{theorem}\label{th:lower}
Suppose that $r\in \Z_{\geq 1}$ and $m\geq d+\floor{\frac{d}{r}}+\epsilon_{d,r}$ where $\epsilon_{d,r}=0$ if $d/r\in \Z$ and $1$ if $d/r\notin \Z$. Then there exist $\{(M_j,\vb_j)\}_{j=1}^m \subset \R^{d\times r}\times \R^{r}$ which has generalized affine phase retrieval property for $\Rd$.
\end{theorem}
\begin{proof}
We set
\[
T_t:=\{(t-1)r+1,\ldots, tr\},\quad t=1,\ldots,\floor{\frac{d}{r}}
\]
and
\[
T_{\floor{\frac{d}{r}}+1}:=\left\{ r\floor{\frac{d}{r}}+1,\ldots,d\right\}.
\]
Note that if ${d}/{r}$ is an integer, then $T_{\floor{\frac{d}{r}}+1}=\emptyset$.
For $\vx\in \R^d$, set $\vx_{T_t}:=\vx{\mathbb I}_{T_t}$ where ${\mathbb I}_{T_t}$  denotes the indicator function of the set $T_t$ (namely ${\mathbb I}_{T_t}(s)=1$ if $s\in T_t$ and $0$ if $s\notin T_t$). Similarly,  we use $(M_j)_{T_t}\in \R^{r\times r}$ to denote a submatrix of $M_j\in \R^{d\times r}$ with row  indexes in $T_t$.
We assume that $\{(M_j,\vb_j)\}_{j=1}^m$ satisfy the following conditions:
\begin{enumerate}[(i)]
\item The matrix $(M_j)_{T_t}=I_r$  and $M_j\setminus (M_j)_{T_t}$ is a zero matrix for $j=(t-1)(r+1)+1,\ldots,t(r+1)$ and $t=1,\ldots,\floor{d/r}$, where $I_r\in \R^{r\times r}$ is the identity matrix.
\item
 Set $\vb_{(t-1)(r+1)+k}=\vb_k'$ for $k=1,\ldots,r+1, t=1,\ldots, \floor{d/r}$. The vectors $\vb'_1,\ldots,\vb'_{r+1}\in \R^r$ satisfy ${\rm span}\{\vb'_2-\vb'_1,\vb'_3-\vb'_1,\ldots,\vb'_{r+1}-\vb'_1\}=\R^r$.
\end{enumerate}
Then, based on Lemma \ref{lemma in R}, for each $t=1,\ldots,\floor{d/r}$, we can recover $\vx_{T_t}$ from $\norm{M_j^\T \vx+\vb_j}, j=(t-1)(r+1)+1,\ldots,t(r+1)$. Hence, when $d/r\in \Z$, we can recover  $\vx=\vx_{T_1}+\cdots+\vx_{T_{(d/r+1)}}$ from $\|M_j\vx+\vb_j\|_2, j=1,\ldots,m$ where $m=(r+1)\floor{d/r}=d+\floor{d/r}$.

When $d/r$ is not an integer, we need consider the recovery of  $\vx_{T_{\floor{d/r}+1}}$. Note that $\#T_{\floor{d/r}+1}=d-r\floor{d/r}$. Similar as before, we can construct matrix $M_{j}\in \R^{d\times r}, $ and $\vb_j\in \R^r, j=\floor{d/r}(r+1)+1,\ldots,\floor{d/r}+d+1$ so that one can recover $\vx_{T_{\floor{d/r}+1}}$ from $\norm{M_{j}^\T \vx+\vb_j}, j=\floor{d/r}(r+1)+1,\ldots,\floor{d/r}+d+1 $. Combining the measurement matrices above, we obtain the measurement number $m=\floor{d/r} (r+1)+d-r\floor{d/r}+1=d+\floor{d/r}+1$ is sufficient to recover $\vx$ provided $d/r$ is not an integer.
\end{proof}

\begin{remark}
If we take $r=d$ in Theorem \ref{th:lower}, we can construct $m=d+1$ matrices $\{(M_j,\vb_j)\}_{j=1}^m$ so that
$\M_{\A}(\vx)=(\|M_1^*\vx+\vb_1\|_2^2,\ldots,\|M_m^*\vx+\vb_m\|_2^2)$
is injective on $\R^d$. Hence, generalized affine phase retrieval can achieve the lower bound
$m=d+1$ which is presented in Theorem \ref{th:d1}.

\end{remark}

As shown in \cite[ Theorem 2.3]{wang2016generalized}, the measurements matrices which have generalize phase retrieval property is an open set.
The following theorem shows that the set of ${\mathcal A}$ having generalized {\em affine} phase retrieval property is {\em not} an open set in $\R^{m(d\times r)}\times \R^{mr}$. The result shows a difference between generalized phase retrieval and generalized affine phase retrieval.

\begin{theorem}
Let $ r\in \Z_{\geq 1}$ and  $m\geq d+\floor{\frac{d}{r}}+\epsilon_{d,r}$ where $\epsilon_{d,r}=0$ if $d/r\in \Z$ and $1$ if $d/r\notin \Z$.
Then the set of generalized affine phase retrieval $\{(M_1,\vb_1),\ldots (M_m,\vb_m)\} \in \R^{m(d\times r)}\times \R^{mr}$ is not an open set in $\R^{m(d\times r)}\times \R^{mr}$.
\end{theorem}
\begin{proof}
To this end, we only need to find a measurement set $\{(M_1,\vb_1),\ldots,(M_m,\vb_m)\} \in \R^{m(d\times r)}\times \R^{mr}$ which has generalized affine phase retrieval property for $\R^d$, but for any $\epsilon>0$ there exists a small perturbation measurement set  $\{(\widetilde{M}_1,\vb_1),\ldots,(\widetilde{M}_m,\vb_m)\}
 \in \R^{m(d\times r)}\times \R^{mr}$ with $\normf{M_j-\widetilde{M}_j}\le \epsilon$ which is not generalized affine phase retrievable.

 We first consider the case where $r=d$.
 Without loss of generality we only need to consider the case $m=d+1$ (for the case where $m>d+1$, we just take $(M_j,\vb_j)=0$ for $j=d+2,\ldots,m$). Set $M_j:=I_d, j=1,\ldots,d+1,$ and assume that $\vb_1,\ldots,\vb_{d+1}\in \Rd$ satisfy
\[
{\rm span}\{\vb_2-\vb_1,\ldots,\vb_{d+1}-\vb_1\}=\Rd.
 \]
 Here, we also require that  the first entry of $\vb_2, \ldots,\vb_{d+1}\in \R^r$ is zero, i.e., $b_{2,1}=\cdots=b_{d+1,1}=0$. According to Lemma \ref{lemma in R}, the measurement set $\{(M_1,\vb_1),\ldots,(M_{d+1},\vb_{d+1})\} \in \R^{(d+1)(d\times d)}\times \R^{(d+1)d}$ has generalized affine phase retrievable property for $\R^d$.

 We perturb $M_1$ to $\widetilde{M}_1=I_d+\delta b_{1,1}E_{21}$, where $E_{21}$ denotes the matrix with $(2,1)$-th entry being $1$ and all other entries being $0$ and $\delta>0$. Furthermore, we let $\widetilde{M}_j=M_j$ for $j=2,\ldots,d+1$. Then $\{(\widetilde{M}_1,\vb_1),\ldots,(\widetilde{M}_{d+1},\vb_{d+1})\} \in \R^{(d+1)(d\times r)}\times \R^{(d+1)r}$ is not generalized affine phase retrievable. To see this, we let $\vx=(b_{1,1},-1/\delta,0,\ldots,0)^\T$ and $\vy=(-b_{1,1},-1/\delta,0,\ldots,0)^\T$. It is easy to check that
\begin{equation*}
  \norm{\widetilde{M}_j^\T \vx+\vb_j}=\norm{\widetilde{M}_j^\T \vy+\vb_j} \quad j=1,\ldots,d+1.
\end{equation*}
By taking $\delta$ sufficiently small, we will have $\normf{M_j-\widetilde{M}_j}\le \epsilon$, which complete the proof for the case where $r=d$.

We next consider the case where $r\leq d-1$. Similar with the proof of Theorem \ref{th:lower}, we set
\[
T_t:=\{(t-1)r+1,\ldots, tr\},\quad t=1,\ldots,\floor{\frac{d}{r}}
\]
and
\[
T_{\floor{\frac{d}{r}}+1}:=\left\{ r\floor{\frac{d}{r}}+1,\ldots,d\right\}.
\]
For $m= d+\floor{\frac{d}{r}}+\epsilon_{d,r}$, we require that $\{(M_1,\vb_1),\ldots,(M_m,\vb_m)\}$ satisfy the conditions (i) and (ii) in the proof of Theorem \ref{th:lower}. We furthermore require that the first entry of $\vb_2,\ldots,\vb_m$ is $0$, i.e., $b_{2,1}=\cdots=b_{m,1}=0$. Note that $(M_1)_{T_1}=I_r$.  We perturb $(M_1)_{T_1}$ to $(\widetilde{M}_1)_{T_1}=I_r+\delta b_{1,1} E_{21}$ and $\widetilde{M}_j=M_j, j=2,\ldots,m$. Then similar as before $(\widetilde{M}_j,\vb_j)_{j=1}^m$ does not have affine phase retrieval property but we will have $\normf{M_j-\widetilde{M}_j}\le \epsilon$ by taking $\delta$ sufficiently small. We complete the proof for $r\leq d-1$.
\end{proof}

The following theorem shows that if the  measurements number $m\ge 2d$, then a generic $\{(M_1,\vb_1),\ldots,(M_m,\vb_m)\} \in \R^{m(d\times r)}\times \R^{mr}$ has generalized affine phase retrieval property for $\Rd$.
\begin{theorem} \label{generic measurements in R}
Let  $m\ge 2d$ and $r\in \Z_{\geq 1}$. Then a generic ${\mathcal A}=\{(M_1,\vb_1),\ldots,(M_m,\vb_m)\} \in \R^{m(d\times r)}\times \R^{mr}$ has generalized affine phase retrieval property for $\Rd$.
\end{theorem}

To prove this theorem, we introduce some notations and a  lemma.  First, recall that
\[ y_j=\norm{M_j^* \vx+\vb_j}^2=\x^* A_j \x, \; j=1,\ldots,m ,
\]
where
\[ \x=\left(\begin{array}{c}
              \vx \\
              1
            \end{array}
\right) \qquad \text{and} \qquad A_j=\left(
                                     \begin{array}{cc}
                                       M_jM_j^* & M_j\vb_j\\
                                       (M_j\vb_j)^* & \vb_j^* \vb_j \\
                                     \end{array}
                                   \right).
\]
Thus, the map $\M_{\A}$ can be rewritten as
\begin{eqnarray*}
  \M_{\A}(\vx) &:=& (\norm{M_1^* \vx+\vb_1}^2,\ldots,\norm{M_m^* \vx+\vb_m}^2) \\
   &=& (\tr(A_1\x\x^*),\ldots,\tr(A_m\x\x^*)).
\end{eqnarray*}
For $\{(M_1,\vb_1),\ldots,(M_m,\vb_m)\} \in \C^{m(d\times r)}\times \C^{mr}$, we define the map $\mathbf{T}:\C^{(d+1)\times (d+1)}\rightarrow \C^m$  by
\begin{equation} \label{the map T}
\mathbf{T}(Q):= \left(\tr(A_1^*Q),\ldots, \tr(A_m^*Q)\right).
\end{equation}

\begin{lemma}\label{equ:generic R}
Suppose that $r\in \Z_{\geq 1}$.  Then ${\mathcal A}=\{(M_1,\vb_1),\ldots,(M_m,\vb_m)\} \in \R^{m(d\times r)}\times \R^{mr}$ is not generalized affine phase retrievable if and only if there exists nonzero $Q\in \R^{(d+1)\times (d+1)}$ satisfies
\begin{equation} \label{condition of Q}
\begin{array}{l}
  Q^\T=Q,\quad Q_{d+1,d+1}=0,\quad \rank(Q)\le 2, \\
  \mathbf{T}(Q)=0, \quad Q_{1,d+1}^2+\cdots+ Q_{d,d+1}^2=1.
\end{array}
\end{equation}
\end{lemma}
\begin{proof}
 Assume that $\A$ is not generalized affine phase retrievable, and then there exist $\vx,\vy \in \R^d$ with $\vx\neq \vy $ such that $\M_\A(\vx)=\M_\A(\vy)$. It implies that
 \[
 \mathbf{T}(\x\x^\T-\y\y^\T)=0,
 \]
 where
 \[
 \x=\left(\begin{array}{c}
              \vx \\
              1
            \end{array}
            \right) ,\qquad
 \y=\left(\begin{array}{c}
              {\mathbf y} \\
              1
            \end{array}
            \right) .
 \]
 Take $Q:=\lambda(\x\x^\T-\y\y^\T)$ where $\lambda=1/\norm{\vx-\vy}^2 \in \R$ is a constant.  Then $Q$ is a nonzero matrix which satisfies (\ref{condition of Q}). \\

 We next assume there exists a nonzero $Q_0$ satisfies (\ref{condition of Q}). According to the spectral decomposition theorem,  we have
 \[ Q_0=\lambda_1 \uu\uu^\T-\lambda_2\uv\uv^\T
 \]
 where $\lambda_1,\lambda_2 \in \R$ and $\uu,\uv $ are normalized orthogonal vectors in $\R^{d+1}$. Since $(Q_0)_{d+1,d+1}=0$, which gives that
 \[ \lambda_1\tilde{u}_{d+1}^2-\lambda_2\tilde{v}_{d+1}^2=0.
 \]
Thus $\lambda_1$ and $\lambda_2$ have the same sign. We claim that $\lambda_1\lambda_2\neq 0$ and $\tilde{u}_{d+1}\tilde{v}_{d+1}\neq 0$. Indeed, if $\lambda_2=0$, then $\tilde{u}_{d+1}=0$. Hence, we obtain $(Q_0)_{1,d+1}=\cdots=(Q_0)_{d+1,d+1}=0$ which contradicts with (\ref{condition of Q}). So, $\lambda_2 \neq 0$. Similarly, we can show $\lambda_1 \neq 0,\; \tilde{u}_{d+1}\neq 0$ and $\tilde{v}_{d+1}\neq 0$. We take $\x:=\uu/\tilde{u}_{d+1}$ and $\y:=\uv/\tilde{v}_{d+1}$, and then $Q_0$ can be rewritten as
\[
Q_0=\lambda_1\tilde{u}_{d+1}^2\x\x^\T-\lambda_2\tilde{v}_{d+1}^2\y\y^\T=c(\x\x^\T-\y\y^\T)
\]
where $c=\lambda_1\tilde{u}_{d+1}^2=\lambda_2\tilde{v}_{d+1}^2 \in \R$ is a constant. Since $\mathbf{T}(Q_0)=0$, it gives that $\mathbf{T}(\x\x^\T)=\mathbf{T}(\y\y^\T)$. We write $\x=(\vx,1)^\top$ and $\y=(\vy,1)^\top$ and then $\M_{\A}(\vx)=\M_{\A}(\vy)$ which implies that $\A$ is not generalized affine phase retrievable.
\end{proof}

\begin{proof}[Proof of Theorem \ref{generic measurements in R} ]
We use $\mathcal{G}_{m,d,r}$ to denote the subset of
\[ (M_1,\vb_1,\ldots,M_m,\vb_m,Q) \in \C^{d\times r} \times \C^r \times \cdots \times \C^{d\times r} \times \C^r\times \C^{(d+1)\times (d+1)},
\]
which satisfies the following property:
\[\begin{array}{l}
  Q^\T=Q,\quad Q_{d+1,d+1}=0,\quad \rank(Q)\le 2, \\
  \mathbf{T}(Q)=0, \quad Q_{1,d+1}^2+\cdots+ Q_{d+1,d+1}^2=1.
\end{array}
\]
The $\mathcal{G}_{m,d,r}$ is a well defined complex affine variety because the defining equations are polynomials in each set of variables. We next consider the dimension of the complex affine variety $\mathcal{G}_{m,d,r}$. To this end, let $\pi_1$  be projections on the first $2m$ coordinates  of $\mathcal{G}_{m,d,r}$, i.e.,
\begin{equation*}
  \pi_1(M_1,\vb_1,\ldots,M_m,\vb_m,Q)=(M_1,\vb_1,\ldots,M_m,\vb_m).
\end{equation*}
Similarly, we can define $\pi_2$ by
\[
\pi_2(M_1,\vb_1,\ldots,M_m,\vb_m,Q)=Q.
\]
We claim that $\pi_2(\mathcal{G}_{m,d,r})=\mathcal{L}_d$ where
\begin{equation*}
  \mathcal{L}_d:=\{Q\in \C^{(d+1)\times (d+1)}: Q^\T=Q,\; Q_{d+1,d+1}=0,\; \rank(Q)\le 2, \; Q_{1,d+1}^2+\cdots+ Q_{d+1,d+1}^2=1\}.
\end{equation*}
Indeed, for any fixed $Q'\in \mathcal{L}_d$, there exist $\{(M'_j,\vb'_j)\}_{j=1}^m \in \C^{d\times r}\times \C^{r}$ satisfying $\mathbf{T}(Q')=0$, because for each $j$ the equation
$\tr((A'_j)^*Q')=0$ is a polynomial for the variables $(M'_j,\vb'_j)$. This implies that $(M'_1,\vb'_1,\ldots,M'_m,\vb'_m,Q')\in \mathcal{G}_{m,d,r}$  and $ \pi_2(M'_1,\vb'_1,\ldots,M'_m,\vb'_m,Q')=Q'$. Thus we have $\pi_2(\mathcal{G}_{m,d,r})=\mathcal{L}_d$. Note that $\mathcal{L}_d \subset \C^{(d+1)\times (d+1)}$ is an affine variety with dimension $2d-1$ and hence $\dimm(\pi_2(\mathcal{G}_{m,d,r}))=2d-1$.

We next consider the dimension of the preimage $\pi_2^{-1}(Q_0)\in \C^{d\times r} \times \C^r \times \cdots \times \C^{d\times r} \times \C^r$ for a fixed nonzero $Q_0\in \mathcal{L}_d$. For each pair $(M_j,\vb_j)\in \C^{d\times r}\times \C^{r}$ , the equation $\tr(A_j^*Q_0)=0$ defines a hypersurface of dimension $dr+r-1$ in $\C^{d\times r} \times \C^r$. Hence, the preimage $\pi_2^{-1}(Q_0)$ has dimension $m(dr+r-1)$. Then, according to \cite[Cor.11.13]{harris2013algebraic}
\begin{eqnarray*}
  \dimm(\mathcal{G}_{m,d,r}) &=& \dimm(\pi_2(\mathcal{G}_{m,d,r}))+\dimm(\pi_2^{-1}(Q_0)) \\
   &=&  m(dr+r-1)+2d-1.
\end{eqnarray*}
If $m\ge 2d$, then
\begin{equation*}
  \dimm(\pi_1(\mathcal{G}_{m,d,r}))\le \dimm(\mathcal{G}_{m,d,r})=m(dr+r-1)+2d-1<m(dr+r).
\end{equation*}
Hence,
\begin{equation*}
  \dimm_\R((\pi_1(\mathcal{G}_{m,d,r}))_\R)\le \dimm(\pi_1(\mathcal{G}_{m,d,r}))<m(dr+r)=\dimm(\R^{m(d\times r)}\times \R^{mr}),
\end{equation*}
which implies that $(\pi_1(\mathcal{G}_{m,d,r}))_\R$ lies in a sub-manifold of $\R^{m(d\times r)}\times \R^{mr}$.
Here, the first inequality follows from \cite{Dan}.
However, Lemma \ref{equ:generic R} implies that $(\pi_1(\mathcal{G}_{m,d,r}))_\R$ contains precisely these $\{(M_1,\vb_1),\ldots,(M_m,\vb_m)\}$ which is not generalized affine phase retrieval for $\R^d$. Hence, we arrive at conclusion.
\end{proof}

\section{Generalized affine phase retrieval for complex signals}
We consider the complex case in this section.
 Then for any $\vx,\vy\in \Cd$, we have
\begin{equation}\label{basic_formula in Cd}
  \|M^*\vx+\vb\|_2^2-  \|M^*\vy+\vb\|_2^2=4\Real\left(\vu^* MM^* \vv+(M\vb)^* \vv\right)
\end{equation}
where $\vu=\frac{1}{2}(\vx+\vy)$ and $\vv=\frac{1}{2}(\vx-\vy)$.
\begin{theorem}\label{equivelant in C}
Suppose that $r\in \Z_{\geq 1}$. Let $\A=\{(M_j,\vb_j)\}_{j=1}^m \subset \C^{d\times r}\times \C^{r}$. Then the followings are equivalent:
\begin{itemize}
  \item[(1)] $\A$ has the generalize affine phase retrieval for $\Cd$.
  \item[(2)] For any $\vu,\vv\in \Cd$ and $\vv\neq 0$, there exists a $j$ with $1\le j\le m$ such that
      $$
      \Real(\vu^* M_jM_j^* \vv+(M_j\vb_j)^* \vv)\neq 0.
      $$
  \item[(3)] Viewing $\M_\A$ as a map $\R^{2d}\rightarrow \R^m$, the real Jacobian of $\M_{\A}(\vx)$ has rank $2d$ for all $\vx \in \R^{2d}$.
\end{itemize}
\end{theorem}
\begin{proof}
(1)$\Leftrightarrow$(2). We first shows that (2) $ \Rightarrow$ (1). We assume that (1) does not hold. Then there exist $\vx\neq \vy$ in $\Cd$ such that $\M_{\A}(\vx)=\M_{\A}(\vy)$. From (\ref{basic_formula in Cd}) for all $j$ we have
\[ \norm{M_j^* \vx +\vb_j}^2-\norm{M_j^* \vy +\vb_j}^2=4\Real(\vu^* M_jM_j^* \vv+(M_j\vb)^* \vv)=0.
\]
Note that $\vv\neq 0$, then we conclude a contradiction with (2), which implies (1) holds. The converse also follows from the similar argument.

(2)$\Leftrightarrow$(3). Note that $M_jM_j^*$ is a Hermitian matrix and we can write $M_jM_j^*=B_j+i C_j$ with $B_j,C_j\in \R^{d\times d}$ and $B_j^\T=B_j, C_j^\T=-C_j$. Let
\begin{equation*}
  F_j=\left(
     \begin{array}{cc}
       B_j & -C_j \\
       C_j & B_j \\
     \end{array}
   \right).
\end{equation*}
 Then for any $\vu=\vu_R+i\vu_I \in \Cd$, we have
\begin{equation*}
  \norm{M_j^* \vu +\vb_j}^2=\uu^\T F_j\uu+2\tilde{\mathbf{c}}_j^\T\uu+\vb_j^*\vb_j.
\end{equation*}
where
\begin{equation*}
  \uu=\left[ \begin{array}{c}
  \vu_R \\
  \vu_I \\
\end{array}\right] \quad \text{and} \quad \tilde{\mathbf{c}}_j=\left[
                                                      \begin{array}{c}
                                                        (M_j\vb_j)_R \\
                                                        (M_j\vb_j)_I \\
                                                      \end{array}
                                                    \right].
\end{equation*}
We can calculate the  real Jacobian $J(\vu)$ of the map $\M_{\A}$ at $\vu\in \Cd$ is exactly
\[ J(\vu)=2[F_1\uu+\tilde{\mathbf{c}}_1,\ldots,F_m \uu+\tilde{\mathbf{c}}_m].
\]
For any $\vv=\vv_R+i\vv_I \in \Cd$, we have
\begin{equation}\label{eq:}
  2\Real(\vu^* M_jM_j^* \vv+(M_j\vb_j)^* \vv)=[\vv_R^\T,\vv_I^\T]J_j(\vu),
\end{equation}
where $J_j(\vu)$ denotes the $j$-column of $J(\vu)$, $\vv_R$ and $\vv_I$ denote the real and image part of $\vv$, respectively. Thus it is clear that (2)  and (3) are equivalent.
\end{proof}

\begin{corollary}
Let $ r \in \Z_{\geq 1}$ and $\A=\{(M_j,\vb_j)\}_{j=1}^m \subset \C^{d\times r}\times \C^{r}$. If
$\A$ has generalized affine phase retrievable property for $\Cd$ then
$m\geq 2d+\floor{d/r}$.
\end{corollary}
\begin{proof}
To this end, we just need show that $\A$ does not have generalized affine phase retrievable property for $\C^d$ provided  $m\leq 2d+\floor{d/r}-1$.
A simple observation is that there exists a $\mathbf{u}_0\in \Cd $ such that $M_j^* \mathbf{u}_0+\vb_j=0$ for all $j=1,\ldots, \floor{d/r}$. Fix $\mathbf{u}_0$, the following system are homogeneous linear equations for the variable $\vv_R,\vv_I\in \Rd$:
\begin{equation}\label{eq:realso}
  \Real((M_jM_j^* \mathbf{u}_0+M_j\vb_j)^* \vv)=0,\quad j=\floor{d/r}+1,\ldots,m.
\end{equation}
Note that those equations have $2d$ real variables $\vv_R,\vv_I$, but the number of equations is at most $2d-1$. It means that (\ref{eq:realso}) must have a nontrivial solution $\vv_0\neq 0$.
 Hence, if $m\leq 2d+\floor{d/r}-1$, then there exist ${\mathbf u}_0, \vv_0\in \Cd$ with
 $\vv_0\neq 0$ so that
   $$
      \Real(\vu_0^* M_jM_j^* \vv_0+(M_j\vb_j)^* \vv_0)= 0, \quad \text{for  all } j=1,\ldots,m
   $$
 which contradicts with (2) in Theorem \ref{equivelant in C}.
\end{proof}

\begin{lemma} \label{lemma in C}
Let $\vz_1,\vz_2\in\C^r$ and suppose that $\vb_1,\ldots,\vb_{2r+1} \in \C^r$
satisfy
\begin{equation}\label{eq:deng}
{\rm span}_{\R}\{\vb_2-\vb_1,\ldots,\vb_{2r+1}-\vb_1\}=\C^r.
\end{equation}
Then $\vz_1=\vz_2$  if $\norm{\vz_1+\vb_j}=\norm{\vz_2+\vb_j}$ for all $j=1,\ldots,2r+1$.
\end{lemma}
\begin{proof}
We set $\vz_R:=\vz_{1,R}-\vz_{2,R} \in \R^r$, $\vz_I:=\vz_{1,I}-\vz_{2,I} \in \R^r$ and $t:=(\norm{\vz_1}^2-\norm{\vz_2}^2)/2$. Then $\norm{\vz_1+\vb_j}=\norm{\vz_2+\vb_j}$ implies that $\vb_{j,R}^\T \vz_R+\vb_{j,I}^\T \vz_I+t=0$ for all $j=1,\ldots,2r+1$. The (\ref{eq:deng}) implies that the rank of the matrix
 \[ A=\left[
      \begin{array}{ccc}
        \vb_{1,R}^\T& \vb_{1,I}^\T  & 1 \\
         \vdots & \vdots& \vdots \\
        \vb_{2r+1,R}^\T & \vb_{2r+1,I}^\T& 1 \\
      \end{array}
    \right]
 \]
 is $2r+1$. And hence $A[\vz_R^\T,\vz_I^\T,t]^\T=0$ has only zero solution which means that $\vz_1=\vz_2$.
\end{proof}

Next, we will show that the bound $m\ge 2d+\floor{d/r}$ is tight provided $r\mid d$.

\begin{theorem} \label{th:lower in C}
Suppose that $m\ge 2d+\floor{d/r}+\epsilon_{d,r}$ where $\epsilon_{d,r}=0$ if $d/r\in \Z$ and $1$ if $d/r\notin \Z$. There exists $\A=\{(M_j,\vb_j)\}_{j=1}^m \subset \C^{d\times r}\times \C^{r}$ which has generalized affine phase retrieval property for $\Cd$.
\end{theorem}
\begin{proof}
We set
\[
T_t:=\{(t-1)r+1,\ldots, tr\},\quad t=1,\ldots,\floor{\frac{d}{r}}
\]
and
\[
T_{\floor{\frac{d}{r}}+1}:=\left\{ r\floor{\frac{d}{r}}+1,\ldots,d\right\}.
\]
We first consider the case where ${d}/{r}$ is an integer with $T_{\floor{\frac{d}{r}}+1}=\emptyset$. Similarly to the real case, for $\vx\in \C^d$, set $\vx_{T_t}:=\vx{\mathbb I}_{T_t}$ where ${\mathbb I}_{T_t}$  denotes the indicator function of the set $T_t$. Let $(M_j)_{T_t}\in \C^{r\times r}$ denote a submatrix of $M_j\in \C^{d\times r}$ with row  indexes in $T_t$.
We assume that $(M_j,\vb_j), j=1,\ldots,m,$ satisfy the following conditions:
\begin{enumerate}[(i)]
\item The matrix $(M_j)_{T_t}=I_r$  and $M_j\setminus (M_j)_{T_t}$ is a zero matrix for $j=(t-1)(2r+1)+1,\ldots,t(2r+1)$ and $t=1,\ldots,\floor{d/r}$, where $I_r $ is ${r\times r}$ the identity matrix.
\item
 Set $\vb_{(t-1)(2r+1)+k}=\vb_k'$ for $k=1,\ldots,2r+1, t=1,\ldots, \floor{d/r}$. The vectors $\vb'_1,\ldots,\vb'_{2r+1}\in \C^r$ satisfy ${\rm span}_{\R}\{\vb'_2-\vb'_1,\vb'_3-\vb'_1,\ldots,\vb'_{2r+1}-\vb'_1\}=\C^r$.
\end{enumerate}
Then based on Lemma \ref{lemma in C}, for each $t=1,\ldots,\floor{d/r}$, we can recover $\vx_{T_t}$ from $\norm{M_j^* \vx+\vb_j}, j=(t-1)(2r+1)+1,\ldots,t(2r+1)$. Hence, when $d/r\in \Z$, we can recover  $\vx=\vx_{T_1}+\cdots+\vx_{T_{(d/r+1)}}$ from $\|M_j^*\vx+\vb_j\|_2, j=1,\ldots,m$ where $m=(2r+1)\floor{d/r}=2d+\floor{d/r}$.

When $d/r$ is not an integer, we need consider the recovery of  $\vx_{T_{\floor{d/r}+1}}$. Note that $\#T_{\floor{d/r}+1}=d-r\floor{d/r}$. Similar as before, we can construct matrix $M_{j}\in \C^{d\times r}, $ and $\vb_j\in \C^r, j=\floor{d/r}(2r+1)+1,\ldots,\floor{d/r}(2r+1)+2d-2r\floor{d/r}+1$ so that one can recover $\vx_{T_{\floor{d/r}+1}}$ from $\norm{M_{j}^* \vx+\vb_j}, j=\floor{d/r}(2r+1)+1,\ldots,\floor{d/r}(2r+1)+2d-2r\floor{d/r}+1 $. Combining the measurement matrices above, we obtain the measurement number $m=\floor{d/r}(2r+1)+2d-2r\floor{d/r}+1=2d+\floor{d/r}+1$ is sufficient to recover $\vx$ provided $d/r$ is not an integer.
\end{proof}

Similar to the real case, the set of ${\mathcal A}\in \C^{m(d\times r)}\times \C^{mr}$ which can do generalized affine phase retrieval is not an open set.
\begin{theorem}
Let $ r\in \Z_{\geq 1}$ and  $m\geq 2d+\floor{\frac{d}{r}}+\epsilon_{d,r}$ where $\epsilon_{d,r}=0$ if $d/r\in \Z$ and $1$ if $d/r\notin \Z$. Then the set of generalized affine phase retrieval $\{(M_1,\vb_1),\ldots,(M_m,\vb_m)\} \in \C^{m(d\times r)}\times \C^{mr}$ is not an open set.
\end{theorem}
\begin{proof}
 We only need to find a measurement set $\{(M_1,\vb_1),\ldots,(M_m,\vb_m)\}\in \C^{m(d\times r)}\times \C^{mr}$ which has generalized affine phase retrieval property for $\C^d$, but for any $\epsilon>0$ there exists a small perturbation measurement set  $\{(\widetilde{M}_1,\vb_1),\ldots,(\widetilde{M}_m,\vb_m)\} \in \C^{m(d\times r)}\times \C^{mr}$ with $\normf{M_j-\widetilde{M}_j}\le \epsilon$ which is not generalized affine phase retrievable.

 We first consider the case where $r=d$.
 Without loss of generality we only need to consider the case $m=2d+1$ (for the case where $m>2d+1$, we just take $(M_j,\vb_j)=({\mathbf 0},{\mathbf 0})$ for $j=2d+2,\ldots,m$). Set $M_j:=I_d, j=1,\ldots,2d+1,$ and
 \begin{equation} \label{the structure of b_j}
 \vb_j=\left\{
        \begin{array}{cl}
          i\ve_j & j=1,\ldots,d \\
          \ve_j & j=d+1,\ldots,2d \\
          0 & j=2d+1 \\
        \end{array}
      \right.,
\end{equation}
where $\{\ve_1,\ldots,\ve_d\}$ is the canonical basis vectors in $\Cd$, i.e. the $j$th entry of $\ve_j$ is $1$ and other entries are $0$. A simple observation is that that $\vb_1,\ldots,\vb_{2d+1}\in \Cd$ satisfy
\[
{\rm span}_{\R}\{\vb_2-\vb_1,\ldots,\vb_{2d+1}-\vb_1\}=\Cd.
 \]
According to Lemma \ref{lemma in R}, the measurement set $\{(M_j,\vb_j)\}_{j=1}^{2d+1}$ has generalized affine phase retrievable property for $\C^d$.

 We perturb $M_1$ to $\widetilde{M}_1=I_d+i\delta E_{12}-i\delta E_{21}$, where $E_{12}$ denotes the matrix with $(1,2)$-th entry being $1$ and all other entries being $0$ and $\delta>0$. Furthermore, we let $\widetilde{M}_j=M_j$ for $j=2,\ldots,2d+1$. Then $\{(\widetilde{M}_j,\vb_j)\}_{j=1}^m \subset \C^{d\times r}\times \C^{r}$ is not generalized affine phase retrievable. To see this, we let $\vx=(i,-\frac{1}{2\delta},0,\ldots,0)^\T$ and $\vy=(-i,-\frac{1}{2\delta},0,\ldots,0)^\T$. It is easy to check that
\begin{equation*}
  \norm{\widetilde{M}_j^* \vx+\vb_j}=\norm{\widetilde{M}_j^* \vy+\vb_j} \quad j=1,\ldots,2d+1.
\end{equation*}
By taking $\delta$ sufficiently small, we will have $\normf{M_j-\widetilde{M}_j}\le \epsilon$, which complete the proof for the case where $r=d$.

We next consider the case where $r\leq d-1$. Using the notations in Theorem \ref{th:lower in C}, we set
\[
T_t:=\{(t-1)r+1,\ldots, tr\},\quad t=1,\ldots,\floor{\frac{d}{r}}
\]
and
\[
T_{\floor{\frac{d}{r}}+1}:=\left\{ r\floor{\frac{d}{r}}+1,\ldots,d\right\}.
\]
For $m= 2d+\floor{\frac{d}{r}}+\epsilon_{d,r}$, we require that $\{(M_j,\vb_j)\}_{j=1}^m$ satisfy the conditions (i) and (ii) in the proof of Theorem \ref{th:lower in C}, i.e.,
 \begin{enumerate}[(i)]
\item The matrix $(M_j)_{T_t}=I_r$  and $M_j\setminus (M_j)_{T_t}$ is a zero matrix for $j=(t-1)(2r+1)+1,\ldots,t(2r+1)$ and $t=1,\ldots,\floor{d/r}$, where $I_r $ is ${r\times r}$ the identity matrix.
\item
 Set $\vb_{(t-1)(2r+1)+k}=\vb_k'$ for $k=1,\ldots,2r+1, t=1,\ldots, \floor{d/r}$. The vectors $\vb'_1,\ldots,\vb'_{2r+1}\in \C^r$ satisfy ${\rm span}_{\R}\{\vb'_2-\vb'_1,\vb'_3-\vb'_1,\ldots,\vb'_{2r+1}-\vb'_1\}=\C^r$.
\end{enumerate}
  Particularly, we require that $\vb'_1,\ldots,\vb'_{2r+1}\in \C^r$ are similarly defined by (\ref{the structure of b_j}). Note that $(M_1)_{T_1}=I_r$. Similar as before, we perturb $(M_1)_{T_1}$ to $(\widetilde{M}_1)_{T_1}=I_r+i\delta  E_{12}-i\delta E_{21}$ and $\widetilde{M}_j=M_j, j=2,\ldots,m$. Then $\{(\widetilde{M}_j,\vb_j)\}_{j=1}^m$ does not have affine phase retrieval property but we will have $\normf{M_j-\widetilde{M}_j}\le \epsilon$ by taking $\delta$ sufficiently small, which completes the proof for $r\leq d-1$.

\end{proof}

\begin{theorem}\label{generic measurements in C}
Let $ r \in \Z_{\geq 1}$ and  $m\ge 4d-1$. Then a generic $\{(M_1,\vb_1),\ldots,(M_m,\vb_m)\} \in \C^{m(d\times r)}\times \C^{mr}$ has generalized affine phase retrieval property for $\Cd$.
\end{theorem}

To this end, we introduce some lemmas.
\begin{lemma}\label{equ:generic C}
Suppose that $r\in \Z_{\geq 1}$.  Then $\A=\{(M_1,\vb_1),\ldots,(M_m,\vb_m)\} \in \C^{m(d\times r)}\times \C^{mr}$ is not generalized affine phase retrievable if and only if there exists nonzero $Q\in \C^{(d+1)\times (d+1)}$ satisfies
\begin{equation} \label{condition of Q in C}
\begin{array}{l}
  Q^*=Q,\quad Q_{d+1,d+1}=0,\quad \rank(Q)\le 2,\quad \mathbf{T}(Q)=0, \\
   Q_{1,d+1}\cdot Q_{d+1,1}+\cdots+ Q_{d,d+1}\cdot Q_{d+1,d}=1,
\end{array}
\end{equation}
where the linear operator $\mathbf{T}$ is defined in (\ref{the map T}).
\end{lemma}
The proof of Lemma \ref{equ:generic C} is similar with one of Lemma \ref{equ:generic R}.
We omit the detail here.
To state conveniently, we use $\C_{\text{sym}}^{d\times d} $ to denote the set of symmetric complex $d\times d$ matrices and use $\C_{\text{skew}}^{d\times d} $ to denote  the set of skew-symmetric complex $d\times d$ matrices.

\begin{definition}
Let $\mathcal{G}_{m,d,r}$ denote the set of $(U_1,\mathbf{c}_1,V_1,\mathbf{d}_1,\ldots,U_m,\mathbf{c}_m,V_m,\mathbf{d}_m,X,Y)$
where $U_j,V_j\in \C^{d\times r}, \mathbf{c}_j, \mathbf{d}_j\in \C^r$, $X\in\C_{\mathrm{sym}}^{(d+1)\times(d+1)}, Y\in \C_{\mathrm{skew}}^{(d+1)\times(d+1)}$
 which satisfy  the following properties:
\[\begin{array}{l}
   X_{d+1,d+1}=0,\quad \rank(X+iY)\le 2,\quad \innerp{A_j,X+iY}=0,\quad j=1,\ldots,m \\
   (X_{1,d+1}+iY_{1,d+1})(X_{d+1,1}+iY_{d+1,1})+\cdots+ (X_{d,d+1}+iY_{d,d+1})(X_{d+1,d}+iY_{d+1,d})=1,
\end{array}
\]
where
\begin{equation}\label{the definition of A_j}
A_j=\left(
\begin{array}{cc}
M_jM_j^* & M_j\vb_j \\
 (M_j\vb_j)^* & \vb_j^*\vb_j \\
 \end{array}
  \right),
\end{equation}
 $M_j=U_j+iV_j$ and $\vb_j=\mathbf{c}_j+i\mathbf{d}_j$.

\end{definition}

 Recall that $\rank(X+iY)\le 2 $ is equivalent to the vanishing to all $3\times 3$ minors of $X+iY$. Hence, we can view  $\mathcal{G}_{m,d,r}$ as a  complex affine variety. Next, we consider the dimension of $\mathcal{G}_{m,d,r}$.

\begin{lemma}
The complex affine variety $\mathcal{G}_{m,d,r}$ has dimension $(2dr+2r-1)m+4d-2$.
\end{lemma}
\begin{proof}
Let $\mathcal{G}_{m,d,r}^{'}$ be the set of $(U_1,\mathbf{c}_1,V_1,\mathbf{d}_1,\ldots,U_m,\mathbf{c}_m,V_m,\mathbf{d}_m,Q)$
 where $U_j,V_j\in \C^{d\times r}, \mathbf{c}_j, \mathbf{d}_j\in \C^r$, $Q\in \C^{(d+1)\times (d+1)}$ which satisfy
\[\begin{array}{l}
  Q_{d+1,d+1}=0,\quad \rank(Q)\le 2,\quad \innerp{A_j,Q}=0, j=1,\ldots,m \\
   Q_{1,d+1}\cdot Q_{d+1,1}+\cdots+ Q_{d,d+1}\cdot Q_{d+1,d}=1,
\end{array}
\]
where matrices $A_j$ are defined by (\ref{the definition of A_j}). Note that
 $\mathcal{G}_{m,d,r}^{'}$ is a well defined complex affine variety because the defining equations are polynomials in each set of variables. It is clear that $\mathcal{G}_{m,d,r}$ and $\mathcal{G}_{m,d,r}^{'}$ are linear isomorphic since we can identify $\C_{\text{sym}}^{d\times d}\times \C_{\text{skew}}^{d\times d}$ with $\C^{d\times d}$ by the map $(X,Y)\mapsto X+iY=Q$. Indeed,  any complex   matrix $Q$ can be uniquely written as $Q=X+iY$ where $X=(Q+Q^\T)/2$ is a complex symmetric matrix and $Y=(Q-Q^\T)/2i$ is a complex skew-symmetric matrix. Hence, to this end, it is sufficient to consider the dimension of $\mathcal{G}_{m,d,r}^{'}$.

We let $\pi_1$ and $\pi_2$ be projections on the first $4m$ coordinates and the last coordinate of $\mathcal{G}_{m,d,r}^{'}$, respectively,  i.e.,
\[
  \pi_1(U_1,\mathbf{c}_1,V_1,\mathbf{d}_1,\ldots,U_m,\mathbf{c}_m,V_m,\mathbf{d}_m,Q)=(U_1,\mathbf{c}_1,V_1,\mathbf{d}_1,\ldots,U_m,\mathbf{c}_m,V_m,\mathbf{d}_m)
\]
and
\[
  \pi_2(U_1,\mathbf{c}_1,V_1,\mathbf{d}_1,\ldots,U_m,\mathbf{c}_m,V_m,\mathbf{d}_m,Q)=Q.
\]
We claim that $\pi_2(\mathcal{G}_{m,d,r})=\mathcal{L}_d$ where
\begin{equation*}
  \mathcal{L}_d:=\{Q\in \C^{(d+1)\times (d+1)}: Q_{d+1,d+1}=0,\; \rank(Q)\le 2, \; Q_{1,d+1}\cdot Q_{d+1,1}+\cdots+ Q_{d,d+1}\cdot Q_{d+1,d}=1\}.
\end{equation*}
Indeed, for any fixed $Q'\in \mathcal{L}_d$, there exists $\{(U'_j,\mathbf{c}'_j,V'_j,\mathbf{d}'_j)\}_{j=1}^m \in \C^{d\times r}\times \C^{r}\times\C^{d\times r}\times \C^{r}$ satisfying $\innerp{A'_j,Q'}=0, j=1,\ldots,m$. It is because that for each $j$ the equation
$\innerp{A'_j,Q'}=0$ is a polynomial which only contain variables $(U'_j,\mathbf{c}'_j,V'_j,\mathbf{d}'_j)$. Thus we have $\pi_2(\mathcal{G}'_{m,d,r})=\mathcal{L}_d$. Note that $\mathcal{L}_d \subset \C^{(d+1)\times (d+1)}$ is an affine variety with dimension $4d-2$ and hence $\dimm(\pi_2(\mathcal{G}'_{m,d,r}))=4d-2$.

We next consider the dimension of the preimage $\pi_2^{-1}(Q_0)$ for a fixed nonzero $Q_0\in \mathcal{L}_d$. For each pair $(U_j,\mathbf{c}_j,V_j,\mathbf{d}_j)$ , the equation $\innerp{A_j,Q_0}=0$ defines a hypersurface of dimension $2dr+2r-1$ in $\C^{d\times r}\times \C^{r}\times\C^{d\times r}\times \C^{r}$. Hence, the preimage $\pi_2^{-1}(Q_0)$ has dimension $m(2dr+2r-1)$. Then, according to \cite[Cor.11.13]{harris2013algebraic}
\begin{eqnarray*}
  \dimm(\mathcal{G}_{m,d,r})=\dimm(\mathcal{G}'_{m,d,r}) &=& \dimm(\pi_2(\mathcal{G}'_{m,d,r}))+\dimm(\pi_2^{-1}(Q_0)) \\
   &=&  m(2dr+2r-1)+4d-2.
\end{eqnarray*}
\end{proof}

\begin{proof}[Proof of Theorem \ref{generic measurements in C} ]
For each $(M_j,\vb_j)\in \C^{d\times r} \times \C^r$, we use $U_j,V_j$ and $\mathbf{c}_j,\mathbf{d}_j$ to denote the real and imaginary part of $M_j$ and $\vb_j$, respectively.
By Lemma \ref{equ:generic C}, a tuple of real matrices $\{(U_j,\mathbf{c}_j,V_j,\mathbf{d}_j)\}_{j=1}^m$ for which the corresponding $\{(M_j,\vb_j)\}_{j=1}^m$ does not have generalized affine phase retrieval property gives a point $\{(U_j,\mathbf{c}_j,V_j,\mathbf{d}_j)\}_{j=1}^m\in \pi_1((\mathcal{G}_{m,d,r})_{\R})\subset (\pi_1(\mathcal{G}_{m,d,r}))_{\R}$. A simple observation is that, if   $m\ge 4d-1$, then
\begin{equation*}
   \dimm(\pi_1(\mathcal{G}_{m,d,r}))\le \dimm(\mathcal{G}_{m,d,r})=m(2dr+2r-1)+4d-2<m(2dr+2r).
\end{equation*}
Hence,
\begin{equation*}
  \dimm_\R((\pi_1(\mathcal{G}_{m,d,r}))_\R)\le \dimm(\pi_1(\mathcal{G}_{m,d,r}))<m(2dr+2r)=\dimm(\R^{d\times r}\times \R^r\times \R^{d\times r}\times \R^r).
\end{equation*}
This implies that the set
\begin{equation*}
  \{(M_j,\vb_j)_{j=1}^m \in \C^{d\times r}\times \C^{r}:(M_j,\vb_j)_{j=1}^m \text{ does not have generalized affine phase retrieval property}\}
\end{equation*}
corresponds to a set $\{(U_j,\mathbf{c}_j,V_j,\mathbf{d}_j)\}_{j=1}^m$ which lies in a sub-manifold of $\R^{d\times r}\times \R^r\times \R^{d\times r}\times \R^r$. Hence, we arrive at conclusion.
\end{proof}

\end{document}